\documentclass[conference]{IEEEtran}
\IEEEoverridecommandlockouts
\usepackage{cite}
\usepackage{amsmath,amssymb,amsfonts}
\usepackage[ruled]{algorithm}
\usepackage[noend]{algpseudocode}
\usepackage{graphicx}
\usepackage{textcomp}
\usepackage{xcolor}
\usepackage{url}
\usepackage{theorem}
\newtheorem{theorem}{Theorem}

\newtheorem{proof}{Proof}[section]

\def\BibTeX{{\rm B\kern-.05em{\sc i\kern-.025em b}\kern-.08em
    T\kern-.1667em\lower.7ex\hbox{E}\kern-.125emX}}
\begin{document}

\title{An Efficient Permissioned Blockchain with Provable Reputation Mechanism}

\author{\IEEEauthorblockN{Hongying Chen\IEEEauthorrefmark{1}, Zhaohua Chen\IEEEauthorrefmark{2}, Yukun Cheng\IEEEauthorrefmark{3}, Xiaotie Deng\IEEEauthorrefmark{4},\\ Wenhan Huang\IEEEauthorrefmark{5}, Jichen Li\IEEEauthorrefmark{6}, Hongyi Ling\IEEEauthorrefmark{7}, Mengqian Zhang\IEEEauthorrefmark{8}}
\IEEEauthorblockA{\IEEEauthorrefmark{1}\IEEEauthorrefmark{2}\IEEEauthorrefmark{6}
\IEEEauthorrefmark{7}School of Electronics Engineering and Computer Science, Peking University, Beijing, 100871, China}
\IEEEauthorblockA{\IEEEauthorrefmark{3}School of Business, Suzhou University of Science and Technology, Suzhou, 215009, China}
\IEEEauthorblockA{\IEEEauthorrefmark{4}Center on Frontiers of Computing Studies,Peking University, Beijing, 100871, China}
\IEEEauthorblockA{\IEEEauthorrefmark{5}\IEEEauthorrefmark{8} School of Electronic Information and Electrical Engineering, Shanghai Jiao Tong University, Shanghai, 200240, China}
\IEEEauthorblockA{Email:  \{\IEEEauthorrefmark{1}chenhongyin, \IEEEauthorrefmark{2}chenzhaohua, \IEEEauthorrefmark{4}xiaotie, \IEEEauthorrefmark{6}2001111325, \IEEEauthorrefmark{7}1600012142\}@pku.edu.cn,\\ \IEEEauthorrefmark{3}ykcheng@amss.ac.cn,\{\IEEEauthorrefmark{5}rowdark, \IEEEauthorrefmark{8}mengqian\}@sjtu.edu.cn}
\IEEEauthorblockA{\IEEEauthorrefmark{3}\IEEEauthorrefmark{4}
These authors are corresponding authors.}
}

\maketitle

\begin{abstract}
The design of permissioned blockchains places an access control requirement for members to read, access, and write information over the blockchains. In this paper, we study a hierarchical scenario to include three types of participants: providers, collectors, and governors. To be specific, providers forward transactions, collected from terminals, to collectors; collectors upload received transactions to governors after verifying and labeling them; and governors validate a part of received labeled transactions, pack valid ones into a block, and append a new block on the ledger. Collectors in the hierarchical model play a crucial role in the design: they have connections with both providers and governors, and are responsible for collecting, verifying, and uploading transactions. However, collectors are rational and some of them may behave maliciously (not necessarily for their own benefits). In this paper, we introduce a reputation protocol as a measure of the reliability of collectors in the permissioned blockchain environment. Its objective is to encourage collectors to behave truthfully and, in addition, to reduce the verification cost. The verification cost on provider $p$ is defined as the total number of invalid transactions provided by $p$ and checked by governors. Through theoretical analysis, our protocol with the reputation mechanism has a significant improvement in efficiency. Specifically, the verification loss that governors suffer is proved to be asymptotically $O(\sqrt{T_{total}})$ ($T_{total}$, representing the number of transactions verified by governors and provided by $p$), as long as there exists at least one collector who behaves well. At last, two typical cases where our model can be well applied are also demonstrated. To our best knowledge, our work is the first one to provide an analytical result on a reputation mechanism in permissioned blockchains and the theoretical results imply our protocol achieves high performance.
\end{abstract}

\begin{IEEEkeywords}
Permissioned blockchain, Transaction verification, Reputation mechanism, Hierarchical structure
\end{IEEEkeywords}

\section{Introduction}\label{Introduction}

In recent years, the popularity of permissioned blockchains has increased visibly, much due to the availability of designed permissions to different users on the network. The most significant difference in the permissioned blockchains is a requirement for participants to be identified. 

While most widely known blockchain schemes, including Bitcoin~\cite{nakamoto2019bitcoin}, Ethereum~\cite{buterin2014next}, and Litecoin~\cite{lee2011litecoin}, are permissionless open ecosystems that allow anyone to join as a validator without authorization, permissioned blockchains are relatively closed and only allow authorized participants to access and to play designated roles. Such permissioned requirements provide an extra security layer to unleash the power of accountability and reliability in blockchains for the regulated institutes such as governments and private companies. 


A key benefit in adopting a permissioned system is to tie a participant's trustfulness in the current task to one's long term behalf. In this paper, we do so by proposing a reputation measure to reduce transaction verification costs. In many of Today's blockchains with smart contracts, transaction verification is no longer as simple as in the transaction verification task required for the Bitcoin system, which mainly verifies digital signatures and local $UTXOs$. Early discovery of the more complicated on-chain transactions' incorrectness will save time cost of transporting them all over the blockchain network. Therefore, reducing transaction verification costs has become an essential task for blockchain applications.

We propose a permissioned blockchain model in a hierarchical structure. There are three kinds of nodes: \textit{providers}, \textit{collectors} and \textit{governors} respectively. Providers offer original transactions to collectors. Collectors will verify the transactions from providers and send labeled transactions to governors. Furthermore, governors are responsible for validating transactions, proposing blocks, and maintaining the ledger. Distinctively, governors only need to verify some of the received transactions with collectors' labels on them. Our protocol aims to reduce the number of invalid transactions checked by governors, which harms efficiency.

As with many real-world cases, collectors \textit{overlap} in gathering transactions from providers. In this sense, it is convenient to apply a \textit{reputation mechanism} on the collectors. In our model, the governor maintains a local reputation vector for each collector that reflects the collector's reliability. Our protocol runs in rounds. In each round, a leading governor, who is selected via a \textit{Proof-of-Stake} (PoS) scheme, is in charge of processing transactions and appending a block to the end of the ledger. For each original transaction $tx$ contained in several labeled transactions from collectors, the leader will decide whether to verify it according to corresponding collectors' reputation. If the leader decides not to verify, then transaction $tx$ would be added into $UncheckedList$; otherwise, the leader will add it into $TXList$ or $InvalidList$ once the verification result is ``valid" or ``invalid" respectively. Once a transaction is verified, governors will update corresponding collectors' reputation according to their label on the transaction. Reputation also serves as the incentive mechanism for collectors to get more profits with higher reputations.

Our protocol pays much attention to efficiency. The loss on provider $p$ is defined as the total number of invalid transactions provided by $p$ and checked by governors. Through theoretical analysis, our protocol with the reputation mechanism has a significant improvement in efficiency.To be specific, the verification loss that governors suffer is proved to be asymptotically $O(\sqrt{T_{total}})$ ($T_{total}$ represents the number of transactions verified by governors and provided by $p$), as long as there exists at least one collector who behaves well. 

Applying our protocol in a permissioned blockchain, governors only pack the transactions in $TXList$ to a block to guarantee the correctness of data, since $TXList$ only contains all valid transactions which are verified in the current round. The Merkle tree root of $(InvalidList, UncheckedList)$ is also contained in a block and $(InvalidList, UncheckedList)$ will be broadcast by governors. From the information of $(InvalidList, UncheckedList)$, it is easy for other participants to make clear which transactions have been checked to be invalid or have not been verified. Once a provider $p_i$ finds his correct transaction is unchecked, he can send this transaction again. The theoretical analysis in this work shows that a valid transaction could be verified within $O(u_i)$ rounds, where $u_i$ is the number of collectors that $p_i$ is connected with.

We also demonstrate our model's feasibility by applying it to the case of IoT data collection and horizontal strategic alliances. In a word, our model solves many issues in the design of reputation mechanisms in permissioned blockchains.

The rest of the paper is arranged as the following: Section~\ref{Background} introduces the background of our model and some related works. In Section~\ref{Model} and Section~\ref{Algorithm}, we build our model and describe the main protocol. We give a full analysis of our model in Section~\ref{Analysis} and apply our model to real-world cases in Section~\ref{Applications}. We conclude our work in Section~\ref{Conclusion}.

\section{Background and Related Work}\label{Background}



\subsection{Transaction Verification in Blockchain}
In traditional blockchain systems, transaction verification includes two parts: authentication verification and correctness verification. Authentication verification is to verify whether relevant signatures are correct. And correctness verification is to verify whether the content of this transaction is valid. For example, in BitCoin, the validity of a transaction equals to that the verification of corresponding signatures is correct and the changes of UTXOs are legal.

Since data in blockchain systems are always digital, authentication can be easily verified by digital signatures. However, sometimes it is not easy to judge the validity of the content of a transaction, which may cost a lot. For example, if the content of the transaction is data trading, then verifying the data will consume a lot of computing power and bandwidth. Under these settings, verifying wrong transactions is a huge waste for miners. At the same time, a fully decentralized system will also bring additional burdens caused by repeated verification of transactions. As a result, we need an efficient distributed system to reduce unnecessary cost in verifying transactions.

\subsection{Consensus in Permissioned Blockchains}

Compared with \textit{permissionless blockchains} focusing on public environments where every node can arbitrarily enter or leave the network, \textit{permissioned blockchains} are operated only by known entities. These permissioned systems maintain the identity and status of participants. For example, in \textit{consortium blockchains}, members of a consortium or in business manage a \textit{permissioned blockchain} network. Meanwhile, in \textit{private blockchains}, only one trusted entity is in charge of the chain~\cite{cachin17blockchain}.

Various permissioned blockchain schemes have appeared and already been implemented in practice. Up to v0.6 of Hyperledger Fabric~\cite{urlfabric}, a native PBFT~\cite{castro02practical} is implemented. In the newest version, an execute-order-validate architecture, rather than the traditional order-execute one, was introduced to carry out the modularity of the protocol. Such a design reduces the need for hard-coding the consensus protocol into the platform and enables Fabric to fit into different circumstances~\cite{androulaki2018hyperledger}. For example, besides the traditional PBFT, BFT-SMaRt~\cite{urlsmart} is also realized as a consensus module. The latter is now considered as one of the most advanced implementations of BFT available. The core consensus scheme lying in Tendermint Core~\cite{urltendermint} is a variation of PBFT which also resists $f < n/3$ Byzantine nodes. However, Tendermint’s most momentous difference from PBFT is the continuous rotation of the leader. Namely, the leader is changed after every block. R3 Corda~\cite{urlcorda} realizes a Hashed Directed Acyclic Graph (Hash-DAG) instead of a chain. A transaction is only stored by those nodes who are affected by the transaction, that is, a node only stores a part of the Hash-DAG. When implemented with Raft~\cite{Ongaro2014in}, R3 Corda tolerates half of the nodes' crashing. The protocol is secure with $f<n/3$ nodes acting arbitrarily when implemented with BFT-SMaRt, as stated before. Quorum~\cite{urlquorum} supports a consensus scheme known as QuorumChain which specifies never-crash-nor-subvert block-maker nodes who are permitted to propose blocks. The scheme guarantees safety and liveness with $f<n/3$ arbitrary-fault nodes and one block-maker. (The chain may fork with multiple block-makers.) Other platforms include Ripple~\cite{urlripple}, Multichain~\cite{urlmulti}, etc..

Besides, beyond traditional BFT algorithms, several fresh consensus models were proposed in recent years, which are considered more suitable for permissioned blockchains. \textit{Proof-of-Authority} (PoA)~\cite{PoA2018}, as a new family of BFT algorithms, requires fewer message exchanges hence provides better performance. Meanwhile, a new message-based consensus scheme named as \textit{Proof-of-Authentication} (PoAh)~\cite{PoAh2020} was presented just a few months ago, which is a lightweight and sustainable permissioned blockchain.

\subsection{Reputation Mechanisms}

Reputation in the \textit{Peer-to-Peer} (P2P) networks has been studied for decades. As the name suggests, reputation works as a measure of the peer's reliability and is evaluated according to his historical behaviors. Various topics related to reputation research, such as the systems design~\cite{DBLP:conf/nossdav/GuptaJA03}~\cite{DBLP:journals/tpds/ZhouH07}, security issue~\cite{buchegger2003robust}~\cite{DBLP:journals/csur/HoffmanZN09} and privacy protection~\cite{DBLP:conf/ecweb/KinatederP03}~\cite{DBLP:journals/cn/WangCYGTW16} have been deeply studied in the literature. In recent years, the emergence and development of blockchain technology bring some new thoughts to the research of reputation. The existing works mainly focus on the following two aspects.

First, because of blockchain's immutability and decentralization, some researches utilize it as a wonderful tool to establish the reputation system~\cite{DBLP:conf/icitst/DennisO15}~\cite{dennis2016rep}~\cite{DBLP:conf/sec/SchaubBHB16}. For example,~\cite{DBLP:conf/dasfaa/GaiWDP18} uses the permissioned blockchain to store the transactional history, which is regarded as the reputation evidence. Then all registered participants' reputation can be evaluated distributedly. Instead of adopting traditional consensus protocols like \textit{Proof-of-Work} (PoW) and \textit{Proof-of-Stake} (PoS) in the process of block generation, the proposed system in ~\cite{DBLP:conf/dasfaa/GaiWDP18} presents a new protocol called \textit{Proof-of-Reputation} (PoR). The protocol assigns one of the nodes involved in the current block, who has the highest reputation value, to be the leader of the current round. Participants all believe that the node with a higher reputation value could provide better services. As a result of taking full advantage of the reputation, the PoR protocol performs efficiently in the designed system.

Second, some studies have applied the concept of reputation as an incentive to the blockchain applications. Our work falls into this research area. ~\cite{nojoumian2018incentivizing} proposes a reputation-based framework for blockchain systems which use PoW as their consensus protocol, to avoid dishonest mining strategies. In the framework, each node has a public reputation value, representing how well he has so far performed in the system. When the pool manager sends invitations to miners to form his mining pool for the proof-of-work computation, one's probability to be invited is related to his reputation value. Nodes who are more reputable have a higher chance to be invited into the mining pools and consequently gain more revenue. Since this public reputation system is sustained over time, miners are incentivized to behave honestly to maximize their long-term utility. However, this paper only lists several possible influence factors and analyzes the update of reputation values qualitatively. The concrete forms of update function and probability function are absent.~\cite{DBLP:journals/corr/abs-2001-06778} presents a protocol called CycLedger for the sharding-based blockchain, which introduces the concept of reputation to provide nodes with enough incentive to enter the system. Nodes in the system are required to give opinions on the validity of requested transactions. One's reputation is computed according to the cosine similarity between his opinions and the consensus results. In this way, the node's reputation is a good reflection of the honest computational resources he has contributed to the system. By assigning nodes with more computational resources to high-workload positions, the efficiency is further enhanced. Also, as a higher reputation brings more revenue, the protocol attracts nodes to participate actively and honestly in the system.

\section{Model}\label{Model}

\subsection{Architecture Overview}\label{Architecture Overview}

Since our protocol is implemented on a \textit{permissioned blockchain}, the identities of participants in the network are all maintained by an \textit{Identity Manager} (IM), who is also responsible for recording participants' roles in blockchain system. Meanwhile, it is in charge of providing nodes credentials that are used for authenticating and authorizing. As a default, an IM should contain all standard Public-Key Infrastructure (PKI) methods and play the role of a Certificate Authority (CA). If no otherwise specified, all interactions within the network are authenticated via digital signatures and IM offers each participant this scheme.

\begin{figure}
    \centering
    \includegraphics[scale = 0.35]{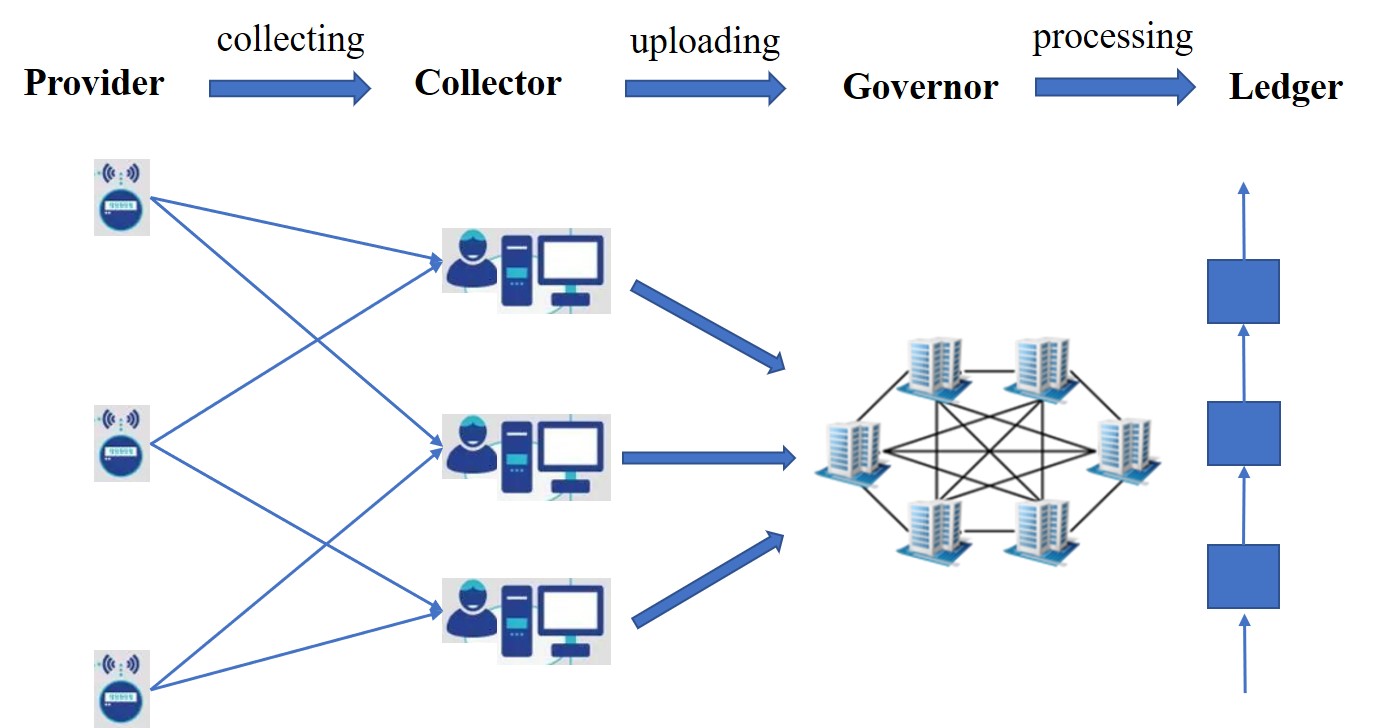}
    \caption{Hierarchical model of the network}
    \label{Hierarchical_Model}
\end{figure}

The hierarchical model studied in this work includes three types of participants: providers, collectors, and governors.

\begin{itemize}
    \item \textit{Providers} would provide original transactions to collectors. Although the correctness of transactions can't be guaranteed, they would sign on transactions together with the \textit{timestamp} so that no collector could forge a transaction.

    \item \textit{Collectors} are responsible for collecting, verifying, and uploading transactions to governors. Specifically, after receiving a signed transaction from a provider, a collector would check it and label it. For each transaction, a collector can label either $+1$ or $-1$, where $+1$ implies that the transaction is considered to be \textit{valid} and $-1$ otherwise. Then the collector would submit the transaction with the corresponding labels and his signature to governors, no matter whether the label is $+1$ or $-1$.

    \item \textit{Governors} are in charge of the block generating for the chain. Because the underlying blockchain is permissioned, we assume that a governor would be elected as the leader in each round by PoS consensus protocol. The leader is responsible for processing transactions and proposing a block. For each original transaction $tx$ contained in several labeled transactions from collectors, the leader will decide whether to verify it according to corresponding collectors' reputations. Finally, the leader will generate a block containing valid transactions and the Merkle tree root of other transactions. In many application scenarios, the cost to verify a transaction may be very high, since a governor shall launch a lot of investigations on this transaction directly from a provider. Once the verification result of a transaction is invalid, it's just a waste of the computation of governors. Therefore, we define the loss on provider $p$  as the total number of invalid transactions provided by $p$ and checked by governors.
\end{itemize}

In this paper, we would like to propose a protocol to collect, verify and pack transactions for a permissioned blockchain environment, by introducing reputation as a measure on the reliability of collectors to encourage collectors to correctly label the transactions. Thereby governors could verify invalid transactions as few as possible, and the verification loss of the whole system can be minimized. The execution of our protocol is partitioned into multi-\textit{rounds}, each containing three phrases, shown in Fig. \ref{Hierarchical_Model}:

\begin{itemize}
    \item \textit{Collecting.} An original transaction $tx$ offered by a single provider is sent to several collectors.

    \item \textit{Uploading.} Collectors verify the received original transactions and label them. These labeled transactions $Tx$ are then uploaded to governors.

    \item \textit{Processing.} In each round, a leader is elected by PoS protocol from all governors. An original transaction $tx$ could be contained in several labeled transactions sent by collectors. The leader chooses a collector according to his reputation, and decides whether to verify $tx$ according to this collector's label. He would discard those unverified transactions. For the verified transactions, on the other hand, the governor would append invalid ones to $UncheckedList$, and append valid ones to $TxList$. Meanwhile, the leader updates the reputations of all collectors according to the verification result by reputation mechanism, which will be introduced in next section. At the end of one round, the leader appends a block, and governors reach consensus on the block as well as the updated reputations of collectors.
\end{itemize}

\subsection{Block Structure}\label{Block Structure and Transaction Validity}

In each round, a new block is appended to the end of a ledger. Each block contains transactions signed by the corresponding providers together with their labels. Note that any block in the ledger is tamper-proof by containing the hash of the previous block in the current one. Moreover, each block has a serial number, such that the blocks in the chain have one-by-one increasing serial numbers. Formally, a block $B$ can be written as
$$B = (s, g_{leader}, TXList = \{tx_1, \cdots, tx_b\}, MT, h)$$
where $s$ is the serial number of block $B$, $g_{leader}$ is the leader in current round, $TXList$ is the list of verified valid transactions, $MT$ is the Merkle Tree root of other collected transactions and $h$ is the hash value of previous block with serial number $s-1$. Here, $b$ is the amount for valid transactions concluded in the block, which is upper bounded by $b_{limit}$.

\subsection{Network Assumption}\label{Network Assumption}

As the underlying system is permissioned, we assume that the system is synchronous, i.e., there is a known upper bound on processing delays, message transmission delays, each node is equipped with a local physical clock and there is an upper bound on the rate at which any local clock deviates from a global real-time clock~\cite{cachin2011introduction}.

The network structures may be diverse and complicated in different actual cases. For simplification, we assume that each governor has connections with all collectors.
The discussion can be easily extended to more general cases.

\subsection{Reputation Mechanism}\label{Reputation Mechanism}

Since all participants are rational, different collectors may provide different labels for a same transaction, due to their different purposes. How to make a judgement on a labeled transaction according to opinions from several collectors is a key issue for governors. Thus a {\em reputation mechanism} is adopted in our protocol, such that the governors keep a reputation list. Each element in reputation list corresponds to one provider $p_i\in P$ and is a $u_i$-length vector, in which the $j$-th component, $j=1,\cdots,u_i$, is just the reputation of the $j$-th collector connected with $p_i$. Intuitively, a collector's reputation is a symbol of his reliability and can help governors make a proper judgement on the labeled transaction from him.


In the reputation mechanism, the probability that one collector is selected by a governor is proportional to his reputation. If the label on the original transaction $tx$ given by this collector is $+1$, then the current leader will verify it; otherwise, he will just append it to $UncheckedList$. Therefore, the reputation mechanism can encourage collectors to behave truthfully, since the higher reputation of one collector is, his opinion on one transaction is more likely accepted. On the other hand, once one collector's malicious behavior is discovered, he must be punished by reducing his reputation.


\subsection{Settings and Notations}\label{Notations}

Based on the above statements, we propose the following notations and some necessary assumptions to formulate the hierarchical setting of the model.
\begin{itemize}
\item The provider set, containing $l$ providers, is denoted by $P=\{p_1, p_2, \cdots, p_l\}$;
    \item The collector set, containing $n$ collectors, is denoted by $C=\{c_1, c_2, \cdots, c_n\}$;
    \item The governor set, containing $m$ governors, is denoted by $G=\{g_1, g_2, \cdots, g_m\}$;
    \item Each provider $p_i\in P$ is connected with $u_i$ collectors;
    \item All the governors maintain a reputation list $\vec{R}_i=(r_i(1), r_i(2), \cdots, r_i(u_i))$ for each provider $p_i\in P$, $i=1,\cdots,l$, where collectors linked to $p_i$ are denoted by $c_i(1), c_i(2), \cdots, c_i(u_i)$, respectively, and $r_i(k)$ reflects the governor $g$'s trust on the collector $c_i(k)$ about the transactions created by $p_i$ and uploaded by $c_i(k)$.
\end{itemize}

There are two types of transactions discussed in this work.
\begin{itemize}
    \item \textit{tx}: unlabeled transaction. Transaction $tx$ is original transaction provided by a provider and has the signature of this provider.
    \item \textit{Tx}: labeled transaction. $Tx = (tx,label,sig_{c_i}(tx,label))$ where $tx$ is $Tx$'s original transaction, $label$ is the label of $tx$ given by collector $c_i$ and $sig_{c_i}(tx,label))$ is the digital signature of $(tx,label)$ generated by $c_i$.
\end{itemize}

Generally, we may call both of them ``transactions". But in some specific algorithms, we shall accurately distinguish them by using notations $tx$ or $Tx$, respectively. In addition, we also apply following functions in subsequent sections.

\begin{itemize}    
    \item $broadcast_{provider}(tx)$. For any provider, he can execute operation $broadcast_{provider}(tx)$ to broadcast an original transaction $tx$ to collectors he is linked to. 

    \item $broadcast_{collector}(Tx)$. For any collector, he can execute operation $broadcast_{collector}(Tx)$ to disseminate a labeled transaction $Tx$ to all governors. 

    \item $broadcast_{governor}(message)$. Each governor can execute operation $broadcast_{governor}(message)$ to broadcast $message$ to all other governors.

    \item $broadcast_{all}(message)$. Everyone in our protocol can execute operation $broadcast_{all}(message)$ to broadcast block $message$ to all linked participants.

    \item $validate_{collector}(tx)$. For any collector, he can execute function $validate_{collector}(tx)$ to check whether transaction $tx$ is valid. If $tx$ is valid, then this function returns $true$, and $false$ otherwise. Since collectors have direct connections with providers, the cost to execute this validation function is not high.

    \item $validate_{governor}(Tx)$. For any governor, he can execute function $validate_{governor}(Tx)$ to check whether a labeled transaction $Tx$ is valid. If $Tx$ is valid, this function returns $true$, and $false$ otherwise. Because governors are farther away from providers, the cost to execute the validation function is much higher than that on collectors.

    \item $parse(Tx)$. Function $parse(Tx)$ returns $(tx,label)$ of a labeled transaction $Tx$.

\end{itemize}

It is worth noting that we've assumed that digital signature is unforgeable. So we ignore the process of verifying signatures for simplicity.

\section{Algorithm}\label{Algorithm}

As stated in Section III, the execution of our protocol is partitioned into multi-rounds, each containing collecting phase, updating phase, and processing phase.

\subsection{Collecting Phase}

In collecting phase, providers offer transactions to collectors.
Specifically, each provider $p_i \in P$ generates some transactions, signs on them together with the timestamp $t$ to constitute $tx$, and then broadcasts them to the $u$ collectors he is in connection with by invoking $broadcast_{provider}(tx)$.

\subsection{Uploading Phase}

In uploading phase, collectors process transactions which are from providers, validate and upload them to governors. Once a collector $c$ receives a transaction $tx$ from provider $p$, he operates $validate_{collector}(tx)$ to check whether $tx$ is valid. The label $label \in \{+1,-1\}$ on $tx$ would be set to $+1$ if $validate_{collector}(tx)$ returns $true$. Otherwise, $label$ will be set to $-1$. Then collectors sign on $(tx,label)$ to generate labeled transaction $Tx$ and broadcast it to governors. Algorithm~\ref{Transactions Uploading} demonstrating this process is shown below.

\begin{algorithm}
    \caption{Transactions Uploading}
    \label{Transactions Uploading}
    \algrenewcommand\algorithmicwhile{\textbf{upon}}
    \begin{algorithmic}[1]
        \Procedure{Transactions\_Uploading}{}\Comment{For collector $c\in C$}
        \While{$\langle\ deliver_{provider}\ |\ p,\ tx\ \rangle$}
            \If{$validate_{collector}(tx) == true$}
            \State $label \gets +1$
            \Else
            \State $label \gets -1$
            \EndIf
            \State $Tx \gets (tx, label, sig_{c}(tx, label))$
            \State \textbf{invoke} $broadcast_{collector}(Tx)$
        \EndWhile
        \EndProcedure
    \end{algorithmic}
\end{algorithm}

\subsection{Processing Phase}

In processing phase, the leader in the current round processes transactions sent by collectors and proposes a new block. The leader is elected by  PoS protocol from all governors in each round, and we'll explain in detail the consensus and leader election algorithm in the fourth part of this phase. The leader may receive several copies of an original transaction $tx$ from collectors, each copy with a label on it. The leader shall make a judgement whether to verify $tx$, according to the reputations of the collectors relative to it. If the leader decides not to verify, then transaction $tx$ would be added into $UncheckedList$; otherwise, the leader will add it into $TXList$ or $InvalidList$ once the verification result is ``valid" or ``invalid", respectively.

Then the governors update the reputations of collectors. Recall that for a governor $g$, the reputation of collectors linked to $p_i$, $\vec{R}_{i}$, is a $u_i$-length vector $\vec{R}_{i} = (r_{i}(1), \cdots, r_i(u_i))$, where we denote $c_i(1), \cdots, c_i(u_i)$ as the $u$ providers that $p_i$ is linked to, and $r_{i}(j)$ is the reputation of $c_i(j)$ on provider $p_i$.

At the end of each round, the leader generates a block and appends it to the ledger. After block generation, all governors would reach a consensus on this new block, and a new round starts. Each round contains four steps as follows.

\subsubsection{Transaction Screening}

In this step, the leader $g_{leader}$ receives transactions in this round and decides whether to verify them one by one. Then each transaction must be appended to $TxList$, $InvalidList$ or $UncheckedLiST$.

As mentioned before, we only consider authenticated transactions since we assume that digital signature is unforgeable. Assume transaction $tx$ is provided by provider $p_i$. Here we give some necessary notations used in this section. Denote the set of collectors linked to $p_i$ by $C_i=\{c_i(1),\cdots,c_i(u_i)\}$, and the reputation list of these collectors on $p_i$ by $\vec{R}_i = (r_i(1), \cdots, r_i(u_i))$.

Suppose a transaction $tx$ proposed by $p_i$ is broadcast to collectors in $C_i$. The leader starts a timing when receiving the first copy of $tx$ with an attached label. Within time of $\Delta$, the leader may receive another $x - 1$ copies of $tx$, each with a label from collectors in $C_i$. Here $x \leq u_i$, since some collectors may simply discard $tx$ or can't finish verifying $tx$ in time.

The leader chooses a collector $c_i(k)$ randomly among $u$ collectors in $C_i$ with the probability proportional to $\exp(\eta r_i(k))$, where $r_i(k)$ is the reputation of $c_i(k)$ on provider $p_i$, and $\eta>0$ is a preset parameter. 
Suppose that $c_i(k)$ is chosen. Noting that when $c_i(k)$ fails to send a copy of $tx$ intentionally or unintentionally, we simply think $c_i(k)$ labels it as $-1$. If the label of $tx$ by $c_i(k)$ is $-1$, the leader would not check it and then append it to $UncheckedList$. If the label is $+1$, then the leader would check $tx$, and append it to $TxList$ once it's proved to be valid, or to $InvalidList$ once it's proved to be invalid.

Algorithm~\ref{Transactions Screening} well illustrates Transactions Screening step.

\begin{algorithm}
    \caption{Transactions Screening}
    \label{Transactions Screening}
    \algrenewcommand\algorithmicwhile{\textbf{upon}}
    \begin{algorithmic}[1]
        \Procedure{Transactions\_Screening}{}\Comment{For governor $g_{leader}$}
        \While{$\langle\ deliver_{collector}\ |\ c,\ Tx\ \rangle$}        
                \State $(tx, label) \gets parse(Tx)$
                \If{$received[tx] == \emptyset$}
                    \State $starttime(tx, \Delta)$\Comment{start timing}
                \EndIf
                \If {$(c, -label) \notin received[tx]$}
                    \State $received[tx]\gets received[tx] \cup \{(c, label)\}$
                \EndIf
        \EndWhile
        \State
        \While{$endtime(tx)$}\Comment{timing is ended}
            \For{all collectors $c_i(k) \in C_i$}
                \State $prob_k\gets \frac{\mathrm{e}^{\eta r_i(k)}}{\sum_{c_i(l) \in C_i} \mathrm{e}^{\eta r_i(l)}}$
            \EndFor
            \State draw a collector $c_i(k) \in C_i$ with probability $prob_k$
            \If{$(c_i(k),+1) \in received[tx]$}
                \State $validbit \gets validate_{governor}(tx)$
                \State $cnt_i \gets cnt_i + 1$
                \State $message \gets (tx, validbit, received[tx], cnt_i)$
                \State \textbf{invoke} $broadcast_{governor}(message)$
                \If{$ validbit == true$}
                    \State $append(TXList, tx)$
                    \State \textbf{invoke} $Reputation\_Updating(tx, +1)$
                \Else
                    \State $append(InvalidList, tx)$
                    \State \textbf{invoke} $Reputation\_Updating(tx, -1)$
                \EndIf
            \Else
                \State $append(UncheckedList, tx)$
            \EndIf
        \EndWhile
        \EndProcedure
    \end{algorithmic}
\end{algorithm}
In Algorithm 2, the leader shall check if collector $c$ has sent a different label on $tx$ to avoid being cheated in line 6. As written in line 14, 15 and 16, the leader would increase $cnt_i$ by $1$ and broadcast $(tx, validbit, received[tx], cnt_i)$ to other governors after verifying $tx$. This broadcasting step informs how collectors label $tx$ and whether $tx$ is valid. Then other governors would trigger $Reputation Updating$ to synchronize the reputation. $cnt_i$ shows the order of $tx$ and helps governors not to omit any reputation updating.

There are two details that haven't be shown in Algorithm 2. First, actually it doesn't matter which governor to validate $tx$ since governors won't tolerate wrong transactions. So the leader can ask any governor to trigger $validate_{governor}$ and get the result back. Second, waiting time $\Delta$ may not be contained in a round, so collectors will send $Tx$ to governors besides the leader, and other governors shall maintain $received[tx]$ too.


\subsubsection{Reputation Updating}


For each transaction that has been validated by the leader, the reputations of collectors that label this transaction correctly will keep the same. In the meantime, those who report the opposite labels will receive a loss of one unit in their reputations. If transaction $tx$ is not verified, then the leader will not conduct the reputation updating operation on collectors who forward it.

Algorithm~\ref{Reputation Updating} demonstrates the description above. Recall that $tx$ is generated by provider $p_i$.

\begin{algorithm}
    \caption{Reputation Updating}
    \label{Reputation Updating}
    \begin{algorithmic}[1]
        \Procedure{Reputation\_Updating}{tx, status}\Comment{For governor $g$ }
        \If{$status == +1$}
            \For{$c_i(k) \in C_i$}
                \If{$(c_i(k), +1) \notin received[tx]$}
                    \State $r_i(k)\gets r_i(k)-1$
                \EndIf
            \EndFor
        \EndIf
        \If{$status == -1$}
            \For{$c_i(k) \in C_i$}
                \If{$(c_i(k), +1) \in received[tx]$}
                    \State $r_i(k)\gets r_i(k)-1$
                \EndIf
            \EndFor
        \EndIf
        \EndProcedure
    \end{algorithmic}
\end{algorithm}

\subsubsection{Block Proposing}
At the end of each round, the leader generates a block $B=(sn, g_{leader}, TXList = \{tx_1, \cdots, tx_b\}, MT, h)$, which contains all valid transactions in $TXList$.

Noting that $sn, h$ can be easily calculated through the former block, and $TXlist, InvalidList, UncheckedList$ are generated in \textit{Transaction Screen} step. The leader computes the Merkle tree root of $(InvalidList, UncheckedList)$ and then gets $MT$. After generating a new block, the leader broadcasts it to all other participants via $broadcast_{all}(B)$. The leader also triggers $broadcast_{all}((InvalidList, UncheckedList))$ to inform transactions which have not been inspected or are illegal. Once receiving the information,
the provider can figure out whether his valid transactions are checked or not, and then decide whether to send it again. All collectors will ignore original transactions that have been proved to be invalid since then. $(InvalidList, UncheckedList)$ may disappear in the distributed system after a certain time since it needn't be permanent.

\subsubsection{Consensus Reaching}

 The leader in each round is elected from all governors. Due to the permissioned system, the leader selection from governors is under the setting of \textit{Proof-of-Stake} (PoS), in which the probability that a governor is selected as the leader is proportional to his stake.

Specifically, suppose governor $g\in G$ owns several units of stake. In practice, the stake may be money or any reliable form of assets. In round $r$, $g$ calculates a hash value with the proof on his own via a \textit{Verifiable Random Function} (VRF)~\cite{MicaliRV99verifiable} scheme for each unit of stake,
and disseminates this hash value with proof to all other governors by invoking $broadcast_{governor}(\cdot)$.
When a governor receives all the hash values and corresponding proofs from other governors, he first validates the proof to check whether the hash value is correct. If the verification is passed, the owner of the stake unit with the least hash value becomes the leading governor of this round.

Such a design is much unsecured in permissionless settings as the leader of a round is predictable. Self-election mechanisms in permissionless blockchains are much more complicated~\cite{GiladHMVZ17}~\cite{KiayiasRDO17}~\cite{DavidGKR18}~\cite{BentovPS16a}. However, in a permissioned environment, we may assume that these governors will not perform malicious behaviors rather than hiding transactions. In other words, there is no motivation for them to damage the chain. In this sense, as the VRF scheme ensures pseudorandomness of the hash value, our leader-electing mechanism guarantees pseudorandomness as well.

In addition, the governors maintain two counters: $cnt_i$ to count how many transactions provided by provider $p_i$ have been verified; and $T_i$ which is initially set to a predetermined integer $T>0$, and doubled every time $cnt_i$ reaches the current value of $T_i$. Recall $p_i$ has $u_i$ collector neighbors, denoted by $\{c_i(1),\cdots,c_i(u_i)\}$, and each of them has reputation $r_i(k)$ on $p_i$. Once $cnt_i$ reaches $T_i$, a constant proportion of the profit gained by executing these transactions will be allotted to collectors according to their reputations. Concretely, collector $c_i(k)$'s revenue would be in proportion to $\exp(\mu r_i(k))$ ($k=1,\cdots,u_i$), where $\mu>0$ is a parameter that can be adjusted at any time. After the revenue allotting, $\{r_i(k)\}_{k}$, the reputations of collectors on provider $p_i$, shall all be reset to 0, and the value of $T_i$ shall be doubled. Since $cnt_i$ is broadcast, governors can easily synchronize reputations even after being reset.
 
Besides, when stakes are transformed among governors, governors related to the transaction should broadcast the signed transaction to all governors. The current leader would then present a stake-transform block at the end of the round.




\section{Analysis}\label{Analysis}

In this section, we shall first analyze the properties our protocol has, and then analyze its efficiency by discussing the complexity in block generation and the reputation mechanism. In addition, we make a discussion on the incentives provided by the reputation mechanism.

\subsection{An Analysis of Consensus Reaching}

Because our protocol is implemented on a permissioned blockchain system, resembling Hyperledger Fabric~\cite{androulaki2018hyperledger}, it has the following properties:

\begin{itemize}
    \item \textit{Agreement}: For any two blocks $B$ and $B'$ retrieved with the same serial number $s$, $B = B'$;

    \item \textit{Chain Integrity}: For any two blocks $B$ and $B'$, if $B$ is retrieved with serial number $s$ while $B'$ is with serial number $s+1$, then the hash value $h'$ contained in $B'$ satisfies $h' = H(B)$, where $H$ is a public collision-resistant hash function (CRHF).

    \item \textit{No Skipping}: Once a block with the serial number $s$ is retrieved by a node, then blocks with all previous serial numbers (i.e., $1, \cdots, s - 1$) are already retrieved before.

    \item \textit{Almost No Creation}: When an \textit{honest} node perceives $tx$ when retrieving a block $B$, then excluding negligible probability of $\lambda$, $tx$ has been broadcast both via $broadcast_{provider}(\cdot)$ and $broadcast_{collector}(\cdot)$ previously. Here $\lambda$ is the security parameter of the protocol.

\end{itemize}

In the consensus reaching process, we use \textit{Verifiable random function} (VRF) scheme among governors to select a leader pseudorandomly. As governors are believed with no motivation to damage the consistency of the chain, blocks are proposed safely. Thus, every governor {\em adopts the same block in each round}, which derives the properties of $Agreement$, $Chain Integrity$ and $No Skipping$ of our protocol. Meanwhile, contributed by the uploading phase, every transaction packed in the block must be labeled $+1$ by at least one collector. At the same time, the collector cannot forge transactions without a provider's public key except with negligible probability of the security parameter $\lambda$. As a result, the property of $Almost No Creation$ can be pledged.

To reach consensus on an ordinary block which contains transactions offered by governors, an $O(b_{limit}m)$ communication complexity emerges.  Meanwhile, to generate a stake-transform block, governors shall broadcast a transaction message to others, which will cause a communication complexity of $O(m^2)$. However, recall that in the permissioned blockchain scheme, the number of governors is much less than the number of providers and collectors. So, the complexity of consensus is not the main efficiency problem that may be suffered from this scheme.

\subsection{Efficiency Analysis for the Reputation Mechanism}

Let us consider one-governor and one-provider case, which can be easily extended to general case, and analyze the efficiency of our protocol.

\begin{theorem}\label{rep}
Consider a group of collectors who have the connections with a same provider $p$. Suppose provider $p$ offers $T$ transactions. Let $\mathcal{L}_T$ be the accumulated expected verification loss on these transactions for the governor, and $\mathcal{S}_T^{min}$ be the accumulated verification loss for the best-behaving collector on these transactions, with any parameter $\eta>0$, we have
\begin{equation}
\mathcal{L}_T-\mathcal{S}_T^{min} \leq \frac{\ln u}{\eta}+\frac{\eta T}{2}, \end{equation}
where $u$ is the number of collectors linked to provider $p$.
\end{theorem}

\begin{proof}
For simplicity, let us denote the reputation of collector $c_i$ on provider $p$ by $r_i$ for any $i=1,\cdots,u$. Our result for one-governor and one-provider case is an extension of the result for the Hedge Algorithm in the learning with expert advice problem~\cite{mohri2018foundations}. So we also follow the potential method to prove it.

The governor deals with all $T$ transactions one by one and updates reputations $\{r_i\}$ of all collectors correspondingly. Let us denote $c_i$'s expected reputation before dealing with the $t$-th transaction by $r_i^{(t)}$.
Then the probability to choose one collector $c_i$ for the $t$-th transaction is $p_i^{(t)}=\frac{\exp{\left(\eta r_i^{(t)}\right)}}{\sum_{j=1}^u\exp{\left(\eta r_j^{(t)}\right)}}$. Denote the change of $c_i$'s reputation by $w_i^{(t)}=r_i^{(t+1)}-r_i^{(t)}$. Note that for the $t$-th transaction,
one of the following events must happen:
\begin{enumerate}
    \item[(1)]
    If collector $c_i$ labels the $t$-th transaction as $-1$, then governor decides not to verify it. So $w_{\ell}^{(t)}=0$ and thus $r_{\ell}^{(t+1)}=r_{\ell}^{(t)}$, for any $\ell=1,\cdots,u$;
    \item[(2)]
    When collector $c_i$ labels the $t$-th transaction as $+1$, then governor verifies the transaction.
    \begin{enumerate}
      \item[(2.1)]
      If the verification result is ``valid", then $w_{\ell}^{(t)}=0$ for each collector $c_{\ell}$, who labels the $t$-th transaction as $+1$; and $w_{\ell}^{(t)}=-1$ for any collector $c_{\ell}$, who labels the $t$-th transaction as $-1$ or fails to report this transaction;
      \item[(2.2)]
      If the verification result is ``invalid", then $w_{\ell}^{(t)}=-1$ for all collectors who label the $t$-th transaction as $+1$, and $w_{\ell}^{(t)}=0$ for others.
    \end{enumerate}
\end{enumerate}
Define potential function $\Phi(t)=\frac{1}{\eta}\ln\sum_{i=1}^u\exp{\left(\eta r_i^{(t)}\right)}$. So,
{\small\begin{align}
    &\phantom{=} \Phi(t+1)-\Phi(t)\nonumber\\
    &= \frac{1}{\eta}\ln\frac{\sum_{i=1}^u\exp{\left(\eta r_i^{(t+1)}\right)}}{\sum_{j=1}^u\exp{\left(\eta r_j^{(t)}\right)}}= \frac{1}{\eta}\ln\sum_{i=1}^u\frac{\exp{\left(\eta (r_i^{(t)}+w_i^{(t)})\right)}}{\sum_{j=1}^u\exp{\left(\eta r_j^{(t)}\right)}}\nonumber\\
    &= \frac{1}{\eta}\ln\sum_{i=1}^up_i^{(t)}\mathrm{e}^{\eta w_i^{(t)}}\leq \frac{1}{\eta}\ln\sum_{i=1}^up_i^{(t)}\left(1+\eta w_i^{(t)}+\frac{\left(\eta w_i^{(t)}\right)^2}{2}\right)\label{ineq:potentialDiff1}\\
    &= \frac{1}{\eta}\ln\left(1+\eta\sum_{i=1}^up_i^{(t)}w_i^{(t)}+\frac{\eta^2\sum_{i=1}^up_i^{(t)}\left(w_i^{(t)}\right)^2}{2}\right)\nonumber\\
    &\leq \frac{1}{\eta}\left(\eta\sum_{i=1}^up_i^{(t)}w_i^{(t)}+\frac{\eta^2\sum_{i=1}^up_i^{(t)}\left(w_i^{(t)}\right)^2}{2}\right)\label{ineq:potentialDiff2}\\
    &= \sum_{i=1}^up_i^{(t)}w_i^{(t)}+\frac{\eta\sum_{i=1}^up_i^{(t)}\left(w_i^{(t)}\right)^2}{2}\nonumber\\
    &= -L(t)+\frac{\eta\sum_{i=1}^up_i^{(t)}\left(w_i^{(t)}\right)^2}{2}\leq -L(t)+\frac{\eta}{2},\label{ineq:potentialDiff3}
\end{align}}
where $L(t)$ is denoted to be the expected verification loss for the $t$-th transaction;
\eqref{ineq:potentialDiff1} is right since $\mathrm{e}^x\leq1+x+\frac{x^2}{2}$ for $x\leq0$, and $w_i^{(t)}\leq0$; \eqref{ineq:potentialDiff2} can be deduced because $\ln(1+x) \leq x$ for $x\geq-1$; \eqref{ineq:potentialDiff3} holds as $\left|w_i^{(t)}\right|\leq1$ and $\sum_{i=1}^up_i^{(t)}=1$. Let us sum \eqref{ineq:potentialDiff3} from $t=1$ to $T$, and notice that $\Phi(1)=\frac{\log u}{\eta}$, then we can get
{\small \begin{equation*}
    \Phi(T+1)-\frac{\ln u}{\eta}\leq-\sum_{t=1}^{T}L(t)+\frac{\eta T}{2}=-\mathcal{L}_T+\frac{\eta T}{2},
\end{equation*}}
which is equivalent to
{\small\begin{equation}
    \mathcal{L}_T+\Phi(T+1)\leq\frac{\ln u}{\eta}+\frac{\eta T}{2}.\label{eq:potential1}
\end{equation}}
Finally, by the definition of $\Phi(T)$ we have
{\small\begin{align}
    \Phi(T+1)
    = \frac{1}{\eta}\ln\sum_{i=1}^u\mathrm{e}^{\eta r_i^{(T+1)}}\geq \frac{1}{\eta}\ln\mathrm{e}^{\eta r_{i^*}^{(T+1)}}
    = r_{i^*}^{(T+1)}= -\mathcal{S}_T^{min},\label{eq:potential2}
\end{align}}
where $i^*$ is the index of the collector with the highest reputation (that is, the collector who behaves best). Combining \eqref{eq:potential1} and \eqref{eq:potential2}, we immediately obtain
{\small\begin{equation*}
    \mathcal{L}_T-\mathcal{S}_T^{min}\leq\frac{\ln u}{\eta}+\frac{\eta T}{2}.
\end{equation*}}
This completes our proof. $\blacksquare$
\end{proof}

Note that for any $T$, we can specify $\eta=\sqrt{\frac{\ln u}{T}}$ and get $\mathcal{L}_T-\mathcal{S}_T^{\mathrm{min}} \leq \frac{3}{2}\sqrt{T\ln u}$. We can double the value of $T$ and reset $\eta$ every time $T$ transactions have been provided by provider $p$. As a result, after $T_{total}=T+2T+\cdots+2^{\ell-1}T=(2^\ell-1)T$ transactions, we have
\begin{align*}
&\phantom{=} \mathcal{L}_{T_{total}}-\mathcal{S}_{T_{total}}^{\mathrm{min}}\\
&\leq \frac{3}{2}\sqrt{T\ln u}+\frac{3}{2}\sqrt{2T\ln u}+\cdots+\frac{3}{2}\sqrt{2^{\ell-1}T\ln u}\\
&=\frac{3}{2}\sqrt{T\ln u}\sum_{k=0}^{\ell-1}(\sqrt{2})^k
= \frac{3}{2(\sqrt{2}-1)}\sqrt{T\ln u}\left((\sqrt{2})^\ell-1\right)\\
&\leq \frac{3}{2(\sqrt{2}-1)}\sqrt{T\ln u}\sqrt{2^{\ell}-1}
= \frac{3\sqrt{\ln u}}{2(\sqrt{2}-1)}\sqrt{(2^{\ell}-1)T}\\
&= O\left(\sqrt{(2^{\ell}-1)T}\right)=O\left(\sqrt{T_{total}}\right).
\end{align*}
As $\ell$ can be arbitrarily large, we thus obtain the $O\left(\sqrt{T_{total}}\right)$ bound of loss for arbitrarily large $T_{total}$.

For the loss of discarding valid transactions, note that we allow providers to repeatedly uploading their transactions as long as they haven't appeared on the chain. As a result, as long as there are honestly behaving collectors, for any valid transaction, there is a positive probability for it to be checked and recorded on the chain every time it is uploaded by the provider. By Law of Large Numbers, it is expected for this transaction to be recorded on the chain in constant rounds of block proposing.

\subsection{Incentive Analysis for the Reputation Mechanism}

A considerable benefit of introducing reputation into our protocol is that we can provide incentives for collectors to behave honestly via the reputation mechanism. In general, as participants in a permissioned system have a sufficient amount of computation resources to complete a given task, we hope a collector's reputation to be a good measurement of his reliability and honesty.

Note that the safety of the protocol is already demonstrated in the previous subsection via the four properties. Typically, to damage the efficiency of our protocol, there are three types of misbehaviors that a collector may commit: (1) to misreport the correct status (\textit{valid} or \textit{invalid}) of a transaction, (2) fail to report a transaction he receives from a provider, and (3) to forge a transaction and report it to the governor.

To simplify the incentive analysis, we consider a situation in which a provider with all collectors he is linked with as well as a governor who is in connection with all these collectors. For an invalid transaction, if a collector misreports it with a valid label, then there is a positive possibility that this transaction is checked by the governor. Once this case happens, the misreporting collector will suffer a reduction in his reputation. For a valid transaction, on the other hand, as long as there are honestly behaving collectors, there is a positive probability that this transaction is checked. Conditioning on it being checked, any collectors who misreport its tag or simply discard it will suffer losses on their reputations.

For those who try to forge a transaction, due to the security of the digital signature scheme, there is only with negligible probability of $\lambda$ to trump up a never-occurred transaction. A malicious collector cannot simply replicate a transaction as well since the transaction is signed together with the timestamp. In all, any fabricating behavior will be detected by the governor except a negligible probability of $\lambda$. Once such behavior is noticed, a loss on the conductor's reputation will severely reduce the revenue of the criminal.

Note that a collector's reputation directly reflects on his revenue. Once a provider $p_i$ generates $T_i$ transactions, collector $c_i(k)$ obtain his revenue proportional to $\exp(\mu r_i(k))$ with $\mu>0$ ($k=1,\cdots,u_i$). Clearly, the more unreliable $c_i(k)$ is, the less $r_i(k)$ will be, as $\mu>0$, the smaller $\exp(\mu r_i(k))$ will be, and less profit $c_i(k)$ will gain. In this sense, our reputation mechanism provides high incentives for collectors to follow the instructions.


\section{Applications and Use Cases}\label{Applications}

The key object of our protocol is to reduce the cost of verifying invalid transactions. The reputation mechanism helps us pick out those collectors with higher reliability, encouraging them to have more incentive to behave truthfully. In this sense, the high efficiency to collect valid transactions and pack them into a block can be realized. In this section, we will show that our model can be well applied in the field of IoT (Internet of Things) data collection and horizontal strategic alliances.

\subsection{IoT Data Collection}

    IoT (Internet of Things) is now everywhere. Millions of front-end devices play an essential role in our life. Data collecting for far-end cloud nodes is a key issue in IoT. Let us consider the situation that cloud nodes need to collect data that meets some certain conditions. It's impossible that all devices send their data to cloud nodes due to the burden of bandwidth.

    One solution is to introduce edge nodes. Edge nodes may be physically close to devices so it costs little for them to obtain data from devices. In addition, edge nodes also can help cloud nodes to filter the data in advance, and then send the filtered data to the cloud nodes. However, edge nodes couldn't be always reliable. They may deliberately misjudge the data or even send a large amount of unqualified data to cloud nodes to paralyze the network and cloud nodes. In addition, the interaction among cloud nodes in such a distributed IoT system also faces difficulties. Besides processing data, cloud nodes need to record the history of data collection and reach a consensus on history. This problem still remains unresolved.

    The protocol we provide in this paper just is a solution for the problems in IoT system, where the devices, edge nodes and cloud nodes can be viewed as providers, collectors and governors, respectively. To be specific, each device obtains data to generate a transaction and then forwards the transaction to edge nodes. The content of a transaction is just the hash value of data together with the signature. Validating a transaction needs to fetch the original data. Edge nodes verify the data in transactions from devices and label them. If an edge node is honest, then it labels a transaction as $+1$, once the data in it meets conditions proposed by cloud nodes; otherwise, the label is $-1$. After labeling the transaction, edge node sends the label and hash value of the data to the leader, who is one cloud node selected in each round and responsible for generating a block. Once receiving hash values and labels from edge nodes, the leader shall make a decision whether to validate the data. If the leader decides to validate based on the labels and reputations of corresponding edge nodes, then he shall fetch the data to verify the hash value and validate the origin data. After validating operation, the leader updates the reputations of all edge nodes according to the validating result. At the end of each round, the leader would generate a block, including the valid transactions and the hash root of other transactions. At the same time, profits are allocated to edge nodes according to their reputations.

    Our protocol in this work can greatly improve the efficiency of data collection for IoT system. Reputation mechanism makes it easy to reduce the cost of invalid data transportation. What's more, reputation mechanism encourages edge nodes to perform honestly. Our scheme also provides higher reliability since cloud nodes can reach consensus and get immutable history, and thus the scenario of IoT system can be extended to other data exchange applications.


\subsection{Horizontal Strategic Alliances}

The strategic alliance is an agreement between two or more companies to pursue mutual benefits~\cite{cravens1993analysis}. Some companies in the same areas establish a strategic alliance for some specific purposes to reduce costs and improve customer service. The strategic alliance above is called a horizontal strategic alliance~\cite{hsa2001eu}. A horizontal alliance is usually managed by a team with members from each company and bounded by an agreement giving an equal risk and opportunity.

In a horizontal strategic alliance, companies unite for a common interest, and there is a need for profit allocation and responsibility for mistakes. The alliance usually needs to collect resources and verify the quality about these resources for their customers. For example, the alliance may need to collect and analyze data from market to optimize the production. However, there is a lot of fraud in market research data, but verifying the correctness of data is expensive. Intermediaries are sometimes introduced since they can usually verify first-hand data more easily. However, intermediaries are not reliable enough that they may deceive companies, which would lead to extra cost for verifying. Therefore, the alliance companies need to find an efficient way to deal with this situation.
 
Our protocol fits well with this problem. For such a three-tier structure of resources, intermediaries and companies, resources serve as providers who generate and provide transactions, intermediaries are responsible for collecting, verifying, and submitting transactions. Meanwhile, companies perform reliant governors in this hierarchy. With the help of reputation mechanism, companies could decrease unnecessary cost of verifying transactions and intermediaries would have incentive to give honest validation for transactions. What's more, the existence of immutable history of blockchain ledger would help companies with their cooperation.

\section{Conclusion}\label{Conclusion}

In this work, we propose a permissioned blockchain model for a hierarchical system, in which there are three types of participants, namely, providers, collectors and governors. They are abstracted from the real world and have their own tasks: providers forward transactions to collectors; collectors label received transactions after checking them, and then broadcast them to governors; governors validate part of labeled transactions, pack valid ones into a block, and append a new block on ledger. Owing to the selfishness of collectors and the overlap of transactions that collectors receive from providers, governors face difficulties that are how to make a judgement on a transaction based on different opinions from different collectors and how to encourage collectors to label the transactions truthfully. To overcome these problems, we design a protocol to collect, verify and pack transactions for the hierarchical permissioned environment, by introducing reputation as a measure on the reliability of collectors. Meanwhile, the collector's revenue is closely related to his reputation. On the one hand, the reputation mechanism used in our protocol can make collectors have incentives to behave honestly. On the other hand, we theoretically prove that our protocol has a significant improvement in efficiency to reduce the verification loss. In a word, with the help of the reputation mechanism, our protocol has high efficiency, great performance, and good incentives. Furthermore, the two applications of IoT Data collection and horizontal strategic alliances show our protocol is more applicable in practice.

\noindent $\mathbf{Acknowledgements}$ Hongying Chen, Zhaohua Chen and Hongyi Ling contributed equally to this work and are joint first authors. Yukun Cheng and Xiaotie Deng are corresponding authors.This research was supported by the National Nature Science Foundation of China (Nos. 11871366, 61761146005, 61632017, 61803279), Qing Lan Project for Young Academic Leaders, Qing Lan Project for Key Teachers.

\bibliographystyle{abbrv}
\bibliography{references}

\begin{thebibliography}{10}

\bibitem{androulaki2018hyperledger}
E.~Androulaki, A.~Barger, V.~Bortnikov, C.~Cachin, K.~Christidis, A.~D. Caro,
  D.~Enyeart, C.~Ferris, G.~Laventman, Y.~Manevich, S.~Muralidharan, C.~Murthy,
  B.~Nguyen, M.~Sethi, G.~Singh, K.~Smith, A.~Sorniotti, C.~Stathakopoulou,
  M.~Vukolic, S.~W. Cocco, and J.~Yellick.
\newblock Hyperledger fabric: a distributed operating system for permissioned
  blockchains.
\newblock In {\em Proceedings of the Thirteenth EuroSys Conference, EuroSys
  2018, Porto, Portugal, April 23-26, 2018}, pages 30:1--30:15, 2018.

\bibitem{PoA2018}
S.~D. Angelis, L.~Aniello, R.~Baldoni, F.~Lombardi, A.~Margheri, and
  V.~Sassone.
\newblock Pbft vs proof-of-authority: Applying the cap theorem to permissioned
  blockchain.
\newblock In {\em Second Italian Conference on Cyber Security, {ITASEC2018}
  2018, Milan, Italy, June 6-9, 2018}, 2018.

\bibitem{BentovPS16a}
I.~Bentov, R.~Pass, and E.~Shi.
\newblock Snow white: Provably secure proofs of stake.
\newblock {\em {IACR} Cryptology ePrint Archive}, 2016:919, 2016.

\bibitem{urlsmart}
BFT-SMaRt.
\newblock \url{https://github.com/bft-smart/library}.

\bibitem{buchegger2003robust}
S.~Buchegger and J.-Y. Le~Boudec.
\newblock A robust reputation system for mobile ad-hoc networks.
\newblock Technical report, 2003.

\bibitem{buterin2014next}
V.~Buterin et~al.
\newblock A next-generation smart contract and decentralized application
  platform.
\newblock {\em white paper}, 3(37), 2014.

\bibitem{cachin2011introduction}
C.~Cachin, R.~Guerraoui, and L.~Rodrigues.
\newblock {\em Introduction to reliable and secure distributed programming}.
\newblock Springer Science \& Business Media, 2nd edition, 2011.

\bibitem{cachin17blockchain}
C.~Cachin and M.~Vukolic.
\newblock Blockchain consensus protocols in the wild.
\newblock {\em CoRR}, abs/1707.01873, 2017.

\bibitem{castro02practical}
M.~Castro and B.~Liskov.
\newblock Practical byzantine fault tolerance and proactive recovery.
\newblock {\em {ACM} Trans. Comput. Syst.}, 20(4):398--461, 2002.

\bibitem{urlcorda}
Corda.
\newblock \url{https://github.com/corda/corda}.

\bibitem{urltendermint}
T.~Core.
\newblock \url{https://github.com/tendermint/tendermint}.

\bibitem{cravens1993analysis}
D.~W. Cravens, S.~H. Shipp, and K.~S. Cravens.
\newblock Analysis of co-operative interorganizational relationships, strategic
  alliance formation, and strategic alliance effectiveness.
\newblock {\em Journal of Strategic Marketing}, 1(1):55--70, 1993.

\bibitem{DavidGKR18}
B.~David, P.~Gazi, A.~Kiayias, and A.~Russell.
\newblock Ouroboros praos: An adaptively-secure, semi-synchronous
  proof-of-stake blockchain.
\newblock In {\em Advances in Cryptology - {EUROCRYPT} 2018 - 37th Annual
  International Conference on the Theory and Applications of Cryptographic
  Techniques, Tel Aviv, Israel, April 29 - May 3, 2018 Proceedings, Part {II}},
  pages 66--98, 2018.

\bibitem{DBLP:conf/icitst/DennisO15}
R.~Dennis and G.~Owen.
\newblock Rep on the block: {A} next generation reputation system based on the
  blockchain.
\newblock In {\em 10th International Conference for Internet Technology and
  Secured Transactions, {ICITST} 2015, London, United Kingdom, December 14-16,
  2015}, pages 131--138. {IEEE}, 2015.

\bibitem{dennis2016rep}
R.~Dennis and G.~Owenson.
\newblock Rep on the roll: a peer to peer reputation system based on a rolling
  blockchain.
\newblock {\em International Journal for Digital Society}, 7(1):1123--1134,
  2016.

\bibitem{urlfabric}
H.~Fabric.
\newblock \url{https://github.com/hyperledger/fabric}.

\bibitem{DBLP:conf/dasfaa/GaiWDP18}
F.~Gai, B.~Wang, W.~Deng, and W.~Peng.
\newblock Proof of reputation: {A} reputation-based consensus protocol for
  peer-to-peer network.
\newblock In J.~Pei, Y.~Manolopoulos, S.~W. Sadiq, and J.~Li, editors, {\em
  Database Systems for Advanced Applications - 23rd International Conference,
  {DASFAA} 2018, Gold Coast, QLD, Australia, May 21-24, 2018, Proceedings, Part
  {II}}, volume 10828 of {\em Lecture Notes in Computer Science}, pages
  666--681. Springer, 2018.

\bibitem{GiladHMVZ17}
Y.~Gilad, R.~Hemo, S.~Micali, G.~Vlachos, and N.~Zeldovich.
\newblock Algorand: Scaling byzantine agreements for cryptocurrencies.
\newblock In {\em Proceedings of the 26th Symposium on Operating Systems
  Principles, Shanghai, China, October 28-31, 2017}, pages 51--68, 2017.

\bibitem{DBLP:conf/nossdav/GuptaJA03}
M.~Gupta, P.~Judge, and M.~H. Ammar.
\newblock A reputation system for peer-to-peer networks.
\newblock In C.~Papadopoulos and K.~C. Almeroth, editors, {\em Network and
  Operating System Support for Digital Audio and Video, 13th International
  Workshop, {NOSSDAV} 2003, Monterey, CA, USA, June 1-3, 2003, Proceedings},
  pages 144--152. {ACM}, 2003.

\bibitem{DBLP:journals/csur/HoffmanZN09}
K.~J. Hoffman, D.~Zage, and C.~Nita{-}Rotaru.
\newblock A survey of attack and defense techniques for reputation systems.
\newblock {\em {ACM} Comput. Surv.}, 42(1):1:1--1:31, 2009.

\bibitem{KiayiasRDO17}
A.~Kiayias, A.~Russell, B.~David, and R.~Oliynykov.
\newblock Ouroboros: {A} provably secure proof-of-stake blockchain protocol.
\newblock In {\em Advances in Cryptology - {CRYPTO} 2017 - 37th Annual
  International Cryptology Conference, Santa Barbara, CA, USA, August 20-24,
  2017, Proceedings, Part {I}}, pages 357--388, 2017.

\bibitem{DBLP:conf/ecweb/KinatederP03}
M.~Kinateder and S.~Pearson.
\newblock A privacy-enhanced peer-to-peer reputation system.
\newblock In K.~Bauknecht, A.~M. Tjoa, and G.~Quirchmayr, editors, {\em
  E-Commerce and Web Technologies, 4th International Conference, EC-Web,
  Prague, Czech Republic, September 2-5, 2003, Proceedings}, volume 2738 of
  {\em Lecture Notes in Computer Science}, pages 206--215. Springer, 2003.

\bibitem{lee2011litecoin}
Litecoin.
\newblock \url{https://litecoin.org/}.

\bibitem{MicaliRV99verifiable}
S.~Micali, M.~O. Rabin, and S.~P. Vadhan.
\newblock Verifiable random functions.
\newblock In {\em 40th Annual Symposium on Foundations of Computer Science,
  {FOCS} '99, 17-18 October, 1999, New York, NY, {USA}}, pages 120--130, 1999.

\bibitem{mohri2018foundations}
M.~Mohri, A.~Rostamizadeh, and A.~Talwalkar.
\newblock {\em Foundations of Machine Learning}.
\newblock The MIT Press, 2nd edition, 2018.

\bibitem{urlmulti}
MultiChain.
\newblock \url{https://github.com/MultiChain/multichain}.

\bibitem{nakamoto2019bitcoin}
S.~Nakamoto.
\newblock Bitcoin: A peer-to-peer electronic cash system.
\newblock Technical report, Manubot, 2019.

\bibitem{nojoumian2018incentivizing}
M.~Nojoumian, A.~Golchubian, L.~Njilla, K.~Kwiat, and C.~Kamhoua.
\newblock Incentivizing blockchain miners to avoid dishonest mining strategies
  by a reputation-based paradigm.
\newblock In {\em Science and Information Conference}, pages 1118--1134.
  Springer, 2018.

\bibitem{Ongaro2014in}
D.~Ongaro and J.~K. Ousterhout.
\newblock In search of an understandable consensus algorithm.
\newblock In {\em 2014 {USENIX} Annual Technical Conference, {USENIX} {ATC}
  '14, Philadelphia, PA, USA, June 19-20, 2014}, pages 305--319, 2014.

\bibitem{PoAh2020}
D.~Puthal, S.~P. Mohanty, V.~P. Yanambaka, and E.~Kougianos.
\newblock Poah: {A} novel consensus algorithm for fast scalable private
  blockchain for large-scale iot frameworks.
\newblock {\em CoRR}, abs/2001.07297, 2020.

\bibitem{urlquorum}
Quorum.
\newblock \url{https://github.com/jpmorganchase/quorum}.

\bibitem{urlripple}
Ripple.
\newblock \url{https://github.com/ripple}.

\bibitem{DBLP:conf/sec/SchaubBHB16}
A.~Schaub, R.~Bazin, O.~Hasan, and L.~Brunie.
\newblock A trustless privacy-preserving reputation system.
\newblock In J.~Hoepman and S.~Katzenbeisser, editors, {\em {ICT} Systems
  Security and Privacy Protection - 31st {IFIP} {TC} 11 International
  Conference, {SEC} 2016, Ghent, Belgium, May 30 - June 1, 2016, Proceedings},
  volume 471 of {\em {IFIP} Advances in Information and Communication
  Technology}, pages 398--411. Springer, 2016.

\bibitem{hsa2001eu}
E.~Union.
\newblock Commission notice: guidelines on the applicability of article 81 of
  the ec treaty to horizontal cooperation agreements.
\newblock
  \url{https://eur-lex.europa.eu/legal-content/EN/TXT/?uri=uriserv:OJ.C_.2001.003.01.0002.01.ENG}.

\bibitem{DBLP:journals/cn/WangCYGTW16}
Y.~Wang, Z.~Cai, G.~Yin, Y.~Gao, X.~Tong, and G.~Wu.
\newblock An incentive mechanism with privacy protection in mobile
  crowdsourcing systems.
\newblock {\em Computer Networks}, 102:157--171, 2016.

\bibitem{DBLP:journals/corr/abs-2001-06778}
M.~Zhang, J.~Li, Z.~Chen, H.~Chen, and X.~Deng.
\newblock Cycledger: {A} scalable and secure parallel protocol for distributed
  ledger via sharding.
\newblock {\em CoRR}, abs/2001.06778, 2020.

\bibitem{DBLP:journals/tpds/ZhouH07}
R.~Zhou and K.~Hwang.
\newblock Powertrust: {A} robust and scalable reputation system for trusted
  peer-to-peer computing.
\newblock {\em {IEEE} Trans. Parallel Distrib. Syst.}, 18(4):460--473, 2007.

\end{thebibliography}

\end{document}